\documentclass[aps,superscriptaddress,twocolumn,showpacs]{revtex4-2}

\usepackage{amsmath}
\usepackage{latexsym}
\usepackage{amssymb}
\usepackage{graphicx}
\usepackage[colorlinks=true, citecolor=blue, urlcolor=blue]{hyperref}
\usepackage{float}
\usepackage{amsfonts}
\usepackage{textcomp}
\usepackage{mathpazo}
\usepackage{blindtext,changes}

\usepackage{bbm}

\usepackage{xcolor}
\definecolor{myurlcolor}{rgb}{0,0,0.4}
\definecolor{mycitecolor}{rgb}{0,0.5,0}
\definecolor{myrefcolor}{rgb}{0.5,0,0}
\usepackage{hyperref}
\hypersetup{colorlinks,
linkcolor=myrefcolor,
citecolor=mycitecolor,
urlcolor=myurlcolor}

\usepackage[draft]{fixme}
\usepackage{amsmath,bbm}
\usepackage{scrextend}
\usepackage{graphicx}
\usepackage{amsfonts}
\usepackage{amssymb}
\usepackage{amsmath,amssymb,amsthm,verbatim,graphicx,bbm}
\usepackage{mathrsfs}
\usepackage{mathtools}
\usepackage{color,xcolor,longtable}

\def\be{\begin{equation}}
\def\ee{\end{equation}}
\def\ben{\begin{eqnarray}}
\def\een{\end{eqnarray}}
\def\eea{\end{array}}
\def\bea{\begin{array}}

\newcommand{\Tr}[1]{\mathrm{Tr}#1}
\newcommand{\bei}{\begin{itemize}}
\newcommand{\eei}{\end{itemize}}
\newcommand{\ket}[1]{|#1\rangle}
\newcommand{\bra}[1]{\langle#1|}

\newcommand{\proj}[1]{\ket{#1}\!\bra{#1}}

\newcommand{\I}{\mathbbm{1}}

\newcommand{\p}{\Vec{p}}

\renewcommand{\emph}[1]{\textbf{#1}}

\makeatletter
\newtheorem*{rep@theorem}{\rep@title}
\newcommand{\newreptheorem}[2]{%
\newenvironment{rep#1}[1]{%
 \def\rep@title{#2 \ref{##1}}%
 \begin{rep@theorem}}%
 {\end{rep@theorem}}}
\makeatother

\theoremstyle{plain}
\newtheorem*{thm*}{Theorem}
\newreptheorem{thm}{Theorem}
\newtheorem{lem}{Lemma}

\theoremstyle{definition}

\theoremstyle{remark}

\usepackage[T1]{fontenc}

\begin{document}


\title{Predictive supremacy of informationally-restricted quantum perceptron}
\author{Shubhayan Sarkar}
\email{shubhayan.sarkar@ug.edu.pl}
\affiliation{Institute of Informatics, Faculty of Mathematics, Physics and Informatics,
University of Gdansk, Wita Stwosza 57, 80-308 Gdansk, Poland}

\begin{abstract}
In the current world, the use of artificial intelligence is penetrating every aspect of human life. The basic element of any artificial intelligence is a digital neuron, called a perceptron, while its quantum analogue is called a quantum perceptron. Here, we introduce a model of perceptron called the informationally-restricted measurement-based perceptron (IMP), where each input is composed of two bits, while at the node, depending on a free input variable, the perceptron decides which bit to evaluate. Additionally, the states transmitted from the input to the node are restricted to a bit (qubit). We establish that under this restriction, the quantum IMP predicts better than a classical IMP. This means that under dimensional restriction of the transmitted states, when both the classical and quantum perceptrons learn the same, the quantum perceptron predicts better than the classical perceptron. For our purpose, we find specific learned values of the perceptron that can display the advantage of a quantum perceptron over its classical counterpart. Restricting to discrete binary inputs, we establish that the observed quantum advantage is universal, that is, for any non-trivial function implementable by both the quantum and classical IMP, one can always find a quantum implementation that outperforms the predictive capability of every classical one. This points to the fact that, given identical learning and resources, a quantum perceptron would predict better than any classical one.
\end{abstract}
\maketitle


\textit{Introduction}--- Artificial intelligence (AI) seeks to build machines that can learn from data, recognise patterns, and make decisions in ways that resemble human cognition. While today’s AI systems are highly sophisticated, their conceptual roots trace back to remarkably simple models of computation inspired by the brain. One of the earliest such models is the neuron proposed by McCulloch and Pitts \cite{McCulloch1943}, which abstracts a biological neuron as a unit that produces a binary output depending on whether its inputs exceed a certain threshold.
Building on this idea, Frank Rosenblatt introduced the perceptron, a model that can learn from examples by adjusting the importance (or weights) assigned to its inputs \cite{Rosenblatt1958}. This marked a crucial step toward adaptive systems, capable not just of processing information but of improving their performance over time. Although simple, the perceptron established the fundamental principle of learning through weighted combinations of inputs, a principle that remains at the heart of modern neural networks.

Quantum perceptron models represent quantum analogues of the classical perceptron, 
aiming to exploit quantum parallelism and amplitude amplification for learning tasks. 
Earlier works proposed a perceptron model based on 
quantum phase estimation to implement a step activation function ~\cite{schuld2014}. 
It was then observed in ~\cite{wiebe2016} that quantum algorithms can achieve 
quadratic improvements in training speed compared to classical methods. 
Later, a unitary quantum perceptron 
featuring a sigmoid-like activation function was constructed in ~\cite{torrontegui2018}. 
Recent works~\cite{springer2022,roget2022} further generalise these approaches, allowing arbitrary 
activation functions and hybrid quantum-classical learning strategies. In all the above models, the usual approach is to simulate the non-linear activation function using some operation in quantum theory. A slightly different approach was considered in ~\cite{benatti2019, calderon2024} 
where the perceptron is implemented by performing local binary outcome measurements on a cluster state, an idea inspired by measurement-based quantum computation \cite{raussendorf2001}. Some other quantum perceptron models and algorithms for learning were proposed in \cite{lloyd2013, otterbach2017unsupervisedmachinelearninghybrid,wiebe2015quantumdeeplearning, wiebe2015, Wan_2017,Schuld_2015}. For a review on the progress of quantum machine learning, refer to \cite{QAIreview, QAIreview2}.

Most models like \cite{wiebe2016} show a quadratic improvement in sample complexity or training steps, but only for specific learning settings, for example, linearly separable data. Moreover, no universal proof exists showing that a quantum neural network learns any function faster than a classical one. Here, we take another approach where, instead of showing the supremacy of a quantum perceptron in learning, we show that a quantum perceptron's supremacy can be established for its predictive power. Consequently, we show that when classical and quantum perceptrons learn the same, the prediction capability of a quantum perceptron can not be reached by any classical counterpart.

Our result establishes the first provable separation in prediction accuracy between classical and quantum perceptrons, even when both implement the same hypothesis and undergo the same training procedure. By working with Rosenblatt’s original perceptron function, we demonstrate that the quantum perceptron is more predictive under identical information-theoretic restrictions. Considering binary inputs, we demonstrate the universality of the observed quantum advantage, that is, for any function implementable in the classical perceptron, one can obtain better predictions if instead one utilised a quantum perceptron given the same restriction on information and same learning. Moreover, the quantum advantage is model uniform, which specifies that for any Boolean function, the same states and measurements can be used to implement the quantum perceptron. This elevates the role of quantum machine learning from toy demonstrations of expressivity to rigorous evidence of predictive superiority, thereby linking quantum information theory, via accessible information and optimal state discrimination, with statistical learning theory, via generalisation error, which evaluates the performance of a perceptron on new, unseen data.

\textit{Classical perceptron---}
We begin by describing the simplest classical perceptron, the basic unit of an artificial neural network [see Fig. \ref{fig1}], also referred to as the McCulloch-Pitts neuron \cite{McCulloch1943}. It is composed of two branches that input the values $x_1,x_2=0,1$ respectively, based on which it computes an activation function $f(.)$ as $ z = f\left( w_1 x_1+w_2x_2 + b \right)$,
where $w_i$ are corresponding weights, and $b$ is the bias term with $z$ taking two values $0,1$. Generally, the function $f(.)$ is non-linear, for instance, in the original Rosenblatt perceptron \cite{Rosenblatt1958}, the activation function is a step function which returns the value $1$ if $w_1 x_1+w_2x_2 + b\geq0$ and $0$ otherwise. In the learning phase, the perceptron is provided with a set of values $\{x_1, x_2, z\}$, also called the learning set, based on which it must assign the values $\{w_1, w_2, b\}$. In the prediction phase, only the values $\{x_1,x_2\}$ is provided, based on which it outputs $z$ depending on the learned $\{w_1, w_2, b\}$. In this work, we restrict to discrete values of $x_1,x_2$ such that each input can take two values $x_i=0,1$. Let us now model the Informationally-restricted Measurement-based Perceptron, or simply, IMP.

\begin{figure*}[t!]
    \centering
    \includegraphics[width=\linewidth]{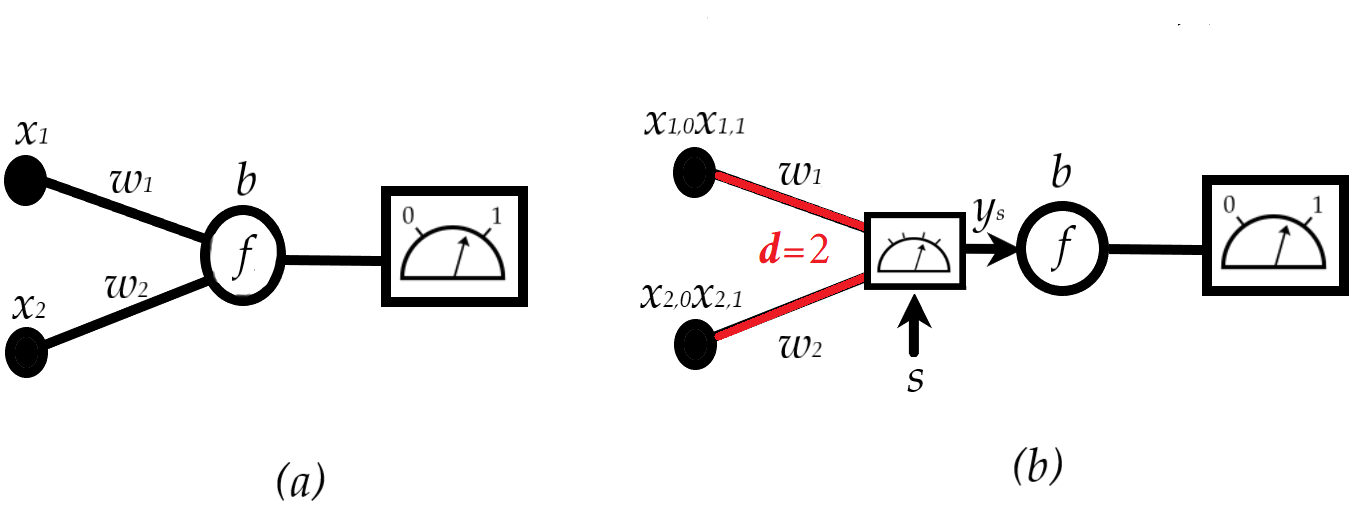}
    \caption{Perceptrons. (a) McCulloch-Pitts neuron with Rosenblatt's weights and bias. Each input $x_i$ is multiplied by a weight $w_i$, and all weighted inputs are summed along with a bias $b$. The neuron then compares this sum with a fixed threshold depending on the function $f$; if the sum is greater than or equal to the threshold, the neuron outputs $1$, otherwise it outputs $0$. (b) Informationally-restricted quantum perceptron. Similar to the McCulloch-Pitts neuron, but each input encodes two bits $x_{i,0}x_{i,1}$ for $i=1,2$. However, the channel connecting the inputs to the node can transmit only a bit (qubit) of information; that is, the maximal dimension of the transmitted system that encodes the inputs is two. At the node, the input $s=0,1$ decides the bit to evaluate. It is important here that $s$ is chosen after the transmitted system reaches the node.}
    \label{fig1}
\end{figure*}

\textit{Informationally-restricted Measurement-based perceptron (IMP)---}
We consider the McCulloch-Pitts neuron and modify it such that each input now sends two bits as $x_i=x_{i,0}x_{i,1}$ such that $x_{i,s}=0,1$ with $s=0,1$. At the node [see Fig. \ref{fig1}], the perceptron is provided with the value $s$, based on which it 
returns the outputs $y_{1,s}y_{2,s}$ where $y_{1,s},y_{2,s}=0,1$. Using the outputs, and depending on each $s$, the perceptron now has to evaluate either of two different functions $f_s$ such that
\begin{equation} z_s=f_s(w_{1,s}y_{1,s}+w_{2,s}y_{2,s}+b_s)
\end{equation}
with their corresponding outputs denoted as $z_s=0,1$. Notice that $z_s$ is the final output of the perceptron while $y_{1,s},y_{2,s}$ are the measurement outputs at the node. Here, we consider that the function along with weights and biases of the perceptron are same for both choices $s$, and thus, $w_{i,0}=w_{i,1}=w_i\ (i=1,2)$ and $b_0=b_1=b$ with $f_s\equiv f$. Without the restriction on the dimension of the transmitted state, this is exactly a standard perceptron with each input $x_i$ sending two bits. However, at the node, based on the choice $s$, the perceptron decides the bit to evaluate. Let us remark here, that the restriction on dimension is important, as any quantum advantage over classical can only be demonstrated under some restriction of resources. Moreover, the choice $s$, can not be revealed before the input bits reach the node. 

Suppose now that the encoded bits are $x'_{1,s},x'_{2,s}$ which gives $f(x'_{1,s},x'_{2,s})=z_s^{\mathrm{cor}}$, thus, $z_s^{\mathrm{cor}}$ is the correct output of the perceptron when the inputs are $x'_{1,s},x'_{2,s}$. If the perceptron predicts perfectly, then $z_s^{\mathrm{pred}}=z_s^{\mathrm{cor}}$, where $z_s^{\mathrm{pred}}$ is the actual output of the perceptron. Moreover, if the measurement outcome at the node is the same as the inputs $y_{i,s}=x_{i,s}'$, then also one would obtain $z_s^{\mathrm{pred}}=z_s^{\mathrm{cor}}$. 
Consequently, for any function $f(.)$, we can define a set of values of $y_{1,s},y_{2,s}$ that return the value $z_s$, that is, $\mathcal{K}_{z_s^{\mathrm{cor}}}=\{y_{1,s},y_{2,s}|f(y_{1,s},y_{2,s})=z_s^{\mathrm{cor}}\}$ for some $z_s^{\mathrm{cor}}=f(x_{1,s}'x_{2,s}')$. Consequently, when the output of the measurement at the node $y_{1,s},y_{2,s}$ belongs to $\mathcal{K}_{z_s^{\mathrm{cor}}}$, the perceptron will output $z_s^{\mathrm{cor}}$. Let us remark that fixing the inputs $x_{1,0},x_{2,1},x_{2,0},x_{2,1}$ gives only a fixed value of $z_s^{\mathrm{cor}}$. This allows us to evaluate the probability that the perceptron correctly outputs $z_s^{\mathrm{cor}}$ as
\begin{eqnarray}\label{psuc}
    p_{\mathrm{cor}}=\frac{1}{\mathcal{N}}\sum_{s}\sum_{x_{1,0},x_{1,1}}\sum_{x_{2,0},x_{2,1}}\sum_{y_{1,s},y_{2,s}\in\mathcal{K}_{z_s^{\mathrm{cor}}}}\qquad\qquad\nonumber\\p(y_{1,s},y_{2,s}|s,x_{1,0},x_{1,1},x_{2,0},x_{2,1})  
\end{eqnarray}
where $\mathcal{N}$ is the number of terms appearing in the summation, which depends on $\mathcal{K}_{z_s^{\mathrm{cor}}}$. Here, $p(y_{1,s},y_{2,s}|s,x_{1,0},x_{2,1},x_{2,0},x_{2,1})$ is the probability of obtaining $y_{1,s},y_{2,s}$ given the inputs $x_{1,0},x_{2,1},x_{2,0},x_{2,1}$ and the query at the node is $s$. 
Recall that $z_s^{\mathrm{cor}}$ is fixed by the inputs $x_{1,0},x_{2,1},x_{2,0},x_{2,1}$. For ease of notation, from here on we denote $\mathbf{x}\equiv x_{1,0},x_{2,1},x_{2,0}$ and $\mathbf{y_s}\equiv y_{1,s},y_{2,s}$. It is clear from the above formula that if the perceptron always outputs $y_{1,s},y_{2,s}\in z_s^{\mathrm{cor}}$, then $p_{\mathrm{cor}}=1$. The maximal value of $p_{\mathrm{cor}}$ \eqref{psuc} attainable using classical strategies will be denoted as $\beta_C$. Since the expression \eqref{psuc} estimates the probability of the perceptron to predict correct outcomes, we call it prediction witness. In short, the basic intuition behind the prediction witness is that, to evaluate the function $f$ correctly, the perceptron needs first to infer the bits sent by the inputs. If the inferred bits correspond to inputs that yield the same output under 
$f$, the prediction is deemed correct. Conversely, the prediction is incorrect if the perceptron identifies a set of bits whose evaluation under $f$ differs from that of the original inputs.

\textit{Quantum-IMP---} Let us now describe the scenario when IMP operates on quantum states and measurements, which will be referred to as Quantum-IMP or Q-IMP. At the inputs, each $x_{1,0},x_{1,1}$ and $x_{2,0},x_{2,1}$ are encoded via quantum states $\rho_{x_{1,0}x_{1,1}}$ and $\rho_{x_{2,0}x_{2,1}}$ respectively and sent to the node. 
Based on the query $s$, the incoming quantum states are measured using some quantum measurement $\{\mathcal{M}^s_{\mathbf{y_s}}\}$ where $\mathcal{M}^s_{\mathbf{y_s}}$ is positive and sums upto identity. The probability $p(\mathbf{y_s}|s,\mathbf{x})$ can be evalutated using the Born rule as
\begin{eqnarray}
p(\mathbf{y_s}|s,\mathbf{x})=\Tr\left(\mathcal{M}^s_{\mathbf{y_{s}}}\rho_{x_{1,0}x_{1,1}}\otimes\rho_{x_{2,0}x_{2,1}}\right)
\end{eqnarray}
The output values are then used to obtain $z_s$ by evaluating $f(.)$. Notice that if one does not restrict the information sent from the inputs, then one can just send the classical labels, which will always lead to correct prediction by the perceptron. The maximal value of $p_{\mathrm{cor}}$ \eqref{psuc} attainable using quantum strategies will be denoted as $\beta_Q$. Let us now proceed to the main results of this work. 

\textit{Quantum advantage of Q-IMP: Single-variable functions---}
Let us begin by noting that for any function $f(.)$, the classical values of $p_{\mathrm{cor}}$ are always obtainable using quantum strategies. Consequently, to observe the supremacy of Q-IMP over IMP, we need to identify functions $f(.)$ such that $\beta_C<\beta_Q$. It is important to note here that to observe a genuine supremacy, we need to consider functions that can be implemented using the classical perceptron. Consequently, functions like XOR can not be considered as they are known not to be implementable using classical perceptrons \cite{Tacchino2019}. This step allows us to put IMP and Q-IMP on the same footing, as they can learn the same values of $w_1,w_2,b$ given the initial training values.  

For our purpose, we consider the step function considered in the Rosenblatt perceptron \cite{Rosenblatt1958} as 
\begin{eqnarray}\label{rosenblatt}
    f(y_{1,s},y_{2,s})=\begin{cases}
        1 & w_1y_{1,s}+w_2y_{2,s}+b\geq0 \\
        0 & w_1y_{1,s}+w_2y_{2,s}+b<0
    \end{cases}
\end{eqnarray}
Let us now state the simplest instance where the supremacy of the Q-IMP can be established.
Suppose after the perceptron is trained, the values of $w_1,w_2,b$ are fixed to be $w_1=1$ and $w_2=b=-1/2$ for $s=0,1$. Plugging these values into the above function \eqref{rosenblatt}, we obtain that when $y_{1,s}=z_s^{\mathrm{pred}}$ for all $s$, that is, the output of the perceptron is only dependent on the first input $y_{1,s}$, and thus, are called single-variable functions. Moreover, we also obtain that $x_{1,s}=z_s^{\mathrm{cor}}$. Thus, we can conclude that in this example $\mathcal{K}_{z_s^{\mathrm{cor}}}$ is given by  $\mathcal{K}_{z_s^{\mathrm{cor}}}=\{y_{1,s},y_{2,s}|y_{1,s}=z_s^{\mathrm{cor}}=x_{1,s}\}$. Consequently, the probability of the perceptron to correctly predict $z_s^{\mathrm{cor}}$ given from \eqref{psuc} can be simplified to
\begin{eqnarray}\label{QRAC}
    p_{\mathrm{cor}}=\frac{1}{8}\sum_{s}\sum_{x_{1,0},x_{1,1}}\sum_{y_{1,s}\atop y_{1,s}=x_{1,s}}p(y_{1,s}|s,x_{1,0},x_{1,1}). 
\end{eqnarray}

This is in fact, the expression of the success probability in the quantum random access codes or Q-RAC in its simplest scenario \cite{Ambainis2002QRAC, Nayak1999QRAC}. It was already shown in \cite{Ambainis2002QRAC} that the maximal value of the above functional \eqref{QRAC} for classical strategies is  $\beta_C=3/4$. In all such scenarios, the best classical strategy usually corresponds to sending either of the bits in each round. Thus, for one value of the choice $s$, the perceptron predicts perfectly, while for the other $s$, the perceptron guesses randomly. Consider now the following quantum states corresponding to inputs $x_{i,s}$ as
\begin{eqnarray}\label{qstates}
\ket{\psi_{x_{i,0},x_{i,1}}}=\sigma_x^{x_{i,0}} \sigma_z^{x_{i,1}}
\left(
\cos\frac{\pi}{8}\,\ket{0}
+
\sin\frac{\pi}{8}\,\ket{1}
\right)
\end{eqnarray}
where $\sigma_{x/z}$ represents the Pauli $X/Z$ operators respectively.
Using the measurements $(j=0,1)$
\begin{eqnarray}\label{mea1}
\mathcal{M}_{i,j}^0&=&\proj{i}\otimes\I_2,\qquad i=0,1\nonumber\\    \mathcal{M}_{0,j}^1&=&\proj{+}\otimes\I_2,\quad \mathcal{M}_{1,j}^1=\proj{-}\otimes\I_2
\end{eqnarray}
where $\ket{\pm}=1/\sqrt{2}(\ket{0}\pm\ket{1})$ and $\I_2$ is the identity matrix of dimension $2$, one can obtain $\beta_Q=1/2(1+1/\sqrt{2})$. Consequently, while the classical IMP can predict with at most 75 \% certainty, the Q-IMP can predict better with an accuracy of close to 85 \%, given the same informational restriction and the same learning.



Let us take another example where the second set of inputs also contribute towards the outcome of the perceptron.


\textit{Quantum advantage of Q-IMP: And-type functions---}
We again consider that the perceptron operates via Rosenblatt's function \eqref{rosenblatt}. Suppose after the perceptron is trained, the values of $w_1,w_2,b$ are fixed to be $w_1=w_2=1$ and $b=-3/2$ for $s=0,1$. Plugging these values into the above function \eqref{rosenblatt}, we obtain that the perceptron implements the And-function, that is, outputs $1$ only when both $y_{1,s}=y_{2,s}=1$ or else $0$ for any $s$. Consequently we have that, $z^{\mathrm{cor}}_s=x_{1,s}.x_{2,s}$ which allows us to conclude that 
$\mathcal{K}_{z_s^{corr}}=\{y_{1,s},y_{2,s}|y_{1,s}.y_{2,s}=z^{\mathrm{cor}}_s=x_{1,s}.x_{2,s}\}$.

This allows us to express the probability of correcting predicting by the perceptron $p_{\mathrm{cor}}$ from \eqref{psuc} as 
\begin{eqnarray}\label{psucand}
    p_{\mathrm{cor}}=\frac{1}{80}\sum_s \sum_{\mathbf{x}}\sum_{\mathbf{y_s}\atop y_{1,s}.y_{2,s}=x_{1,s}.x_{2,s}} p(\mathbf{y_{s}}|s,\mathbf{x})
\end{eqnarray}
The factor $80$ in the denominator of the above functional arises due to the number of individual probability terms in the above functional. 
Let us now evaluate the classical bound $\beta_C$ of the functional \eqref{psucand}.

\begin{lem}
    The classical bound of \eqref{psucand} is $\beta_C=\frac{13}{40}$.
\end{lem}
\begin{proof}
    Considering the simplest classical, one sends the first bit corresponding to $s=0$, and thus we get 
    \begin{eqnarray}
p(\mathbf{y_{0}}|s=0,\mathbf{x})=1\quad \mathrm{for}\ y_{i,0}=x_{i,0},
    \end{eqnarray}
    and $0$ otherwise.
    For $s=1$, the perceptron outputs uniformly accross any input and thus we have
    \begin{eqnarray}
p(\mathbf{y_{1}}|s=1,\mathbf{x})=\frac{1}{4}.
    \end{eqnarray}
    This gives us $\beta_C=\frac{1}{80}(16+10)=\frac{13}{40}.$
\end{proof}
Let us now show that with the quantum states defined in \eqref{qstates} and measurements 
\begin{eqnarray}\label{meaand}
    M_{i,j}^0&=&\proj{i}\otimes\proj{j},\qquad i,j=0,1\nonumber\\
    M_{i,j}^1&=&\proj{+_i}\otimes\proj{+_j},
\end{eqnarray}
where $\ket{+_j}=1/\sqrt{2}(\ket{0}+(-1)^j\ket{1})$ one can obtain a $p_{\mathrm{cor}}$ \eqref{psucand} higher than $\beta_C$.

\begin{lem}
Consider the states \eqref{qstates} and measurements \eqref{meaand}. Using them, we obtain that $\beta_Q=\frac{1}{40}\left(11+2\sqrt{2}\right)$ of \eqref{psucand}.
\end{lem}
\begin{proof}
    Let us construct the probability distribution using the quantum states \eqref{qstates} and measurements \eqref{meaand} as
    \begin{eqnarray}
       p(\mathbf{y_0}|0,\mathbf{x})=\begin{cases}
            \frac{1}{4}(1+\frac{1}{\sqrt{2}})^2\quad y_{1,0}=x_{1,0},y_{2,0}=x_{2,0}\nonumber\\
             \frac{1}{8}\qquad\qquad\quad\  y_{1,0}=x_{1,0},y_{2,0}\ne x_{2,0}\nonumber\\
             \frac{1}{8}\qquad\qquad\quad\  y_{1,0}\ne x_{1,0},y_{2,0}= x_{2,0}\nonumber\\
             \frac{1}{4}(1-\frac{1}{\sqrt{2}})^2\quad y_{1,0}\ne x_{1,0},y_{2,0}\ne x_{2,0}
        \end{cases}.
    \end{eqnarray}
    One can obtain the exact same distribution for $s=1$. Thus, we obtain the value of \eqref{psucand} to be $\beta_Q=\frac{1}{40}\left(11+2\sqrt{2}\right)$. This completes the proof.
   \end{proof}
        

Let us now establish the universality of the quantum advantage.

\textit{Universal quantum advantage for discrete binary inputs---} Here, we demonstrate that for every non-trivial function implementable with both the classical and quantum perceptron when the inputs are binary, the quantum one predicts better. For this purpose, let us first recall all the functions implementable by a Rosenblatt perceptron when the inputs are binary. One can directly infer that, in this setting, the perceptron realises a Boolean mapping of the form $\{0,1\}^2 \to \{0,1\}$. It is a standard result that a Boolean function can be represented by a perceptron if and only if it is linearly separable \cite{minsky1969perceptrons}. Among the $16$ possible Boolean functions of two binary variables, exactly two, namely XOR and XNOR, fail to satisfy linear separability. Therefore, the perceptron is capable of implementing exactly $14$ such functions. The truth tables, along with some possible values of the parameters $w_1,w_2,b$ are given in Appendix A. They are grouped into four classes: (a) Constant (b) Single-valued (c) And/Or, (d) Assymetric. As IMP is based on the same setting as the Rosenblatt perceptron, any function corresponding to choice $s$, implements either of these gates. 

In the above examples, we demonstrated the predictive supremacy of Q-IMP over IMP when (i) the function $f$ is single-valued, (ii) the function is of the And type. Now, it is straightforward to observe that whenever $f$ is constant, then it is trivial, as any perceptron can always predict the outcome with certainty. Consequently, we consider all the possible implementations when $f$ is non-trivial. In Appendix B, we state all the possible non-trivial functions $f$ and find the maximal value attainable in a classical-IMP versus a quantum-IMP. In all such non-trivial implementations, we find that quantum-IMP predicts better than classical IMP. For this purpose, we use the fact that any other function is basically a rearrangement of single or And type functions. 

Since any classical perceptron with binary inputs can only generate these functions, we establish that, when a quantum and classical perceptron implement the same functions and learn the same, the quantum perceptron performs strictly better for any weights, biases or functions. Moreover, it is universal and model uniform, as with the same set of states, (up to ignoring one of them when the perceptron implements single-variable functions) and measurements, one can obtain an advantage in the predictive power of the quantum perceptron over a classical one. One might be tempted to conclude that the advantage is small; however, we are only observing this for one perceptron and with each new layer, the gap is expected to increase.

\textit{Discussions---}
Some functions, such as, XOR can be realised by a single-layer quantum perceptron, whereas a single-layer classical perceptron, restricted to linear threshold operations, cannot \cite{Tacchino2019}. This separation, however, does not reflect a unique power of the perceptron architecture itself, but rather the broader fact that quantum mechanics permits a richer set of operations than classical physics. To claim a genuine perceptron supremacy, both classical and quantum models must be placed on the same footing: identical function, input–output specifications, and layer depth. Under these conditions, a meaningful separation must arise not merely from representability, but from demonstrable differences in learning efficiency or predictive accuracy. We observed above that even when classical and quantum perceptrons can both implement the same function along with the same learning, the quantum perceptron achieves strictly better prediction accuracy under the same informational restriction. This is the only evidence, as of now, of a predictive advantage of a quantum perceptron over any classical one. One can thus conclude that quantum learners extract more predictive information per bit of data access than classical ones. Moreover, this result establishes a rigorous information-theoretic separation in prediction, rather than relying on conventional complexity-theoretic arguments, where apparent advantages often diminish as improved classical algorithms are developed. The separation here is fundamental, which can not be reduced with better classical strategies. Conceptually, this finding reframes the debate: the quantum advantage is not confined to contrived functions or asymptotic speedups, but arises already in the simplest, historically central model of learning.

Several interesting follow-up problems arise from this work. The most interesting among them would be to observe whether the informationally restricted perceptron learns better when using quantum realisations than classical ones. Another important problem in this direction will be to extend the results from discrete binary inputs to the continuous case, which can be modelled by considering an arbitrary number of discrete variables at the inputs. Furthermore, reducing energy consumption has become a central challenge in modern commercial artificial intelligence, as the growing scale and complexity of these systems demand increasingly significant computational resources. It will be intriguing if the restriction on information can be replaced with a restriction on energy, and then one could demonstrate that with the same amount of energy, a quantum perceptron performs better. Here, we demonstrated the quantum advantage in prediction by a single perceptron. It will be highly appealing to analyse this setup considering neural networks, and observe how the gap scales with larger networks.

\textit{Acknowledgements---}
    This project was funded within the National Science Centre, Poland, grant Opus 25, UMO-2023/49/B/ST2/02468.


\bibliography{ref.bib}

\newpage 
\onecolumngrid
\appendix
    \section{ All Implementable functions of a two-bit classical perceptron}

Here, we provide a brief overview of all possible functions that can be implemented using a two-bit classical perceptron. The inputs of any two-bit classical perceptron are given by $(x_1,x_2)$ where $x_1,x_2=0,1$ with output given as $y=f(w_1x_1+w_2x_2+b)$ where $y=1$ if $w_1x_1+w_2x_2+b\ge 0$ and $0$ otherwise. Consequently, one can straightaway conclude that the perceptron in this scenario implements a Boolean function: $\{0,1\}^2\rightarrow\{0,1\}$.
It is well-known that a Boolean function is implementable by a perceptron if and only if it is linearly separable \cite{minsky1969perceptrons}. Since there are $16$ Boolean functions
on two bits and exactly two of them (XOR and XNOR) are not linearly separable,
the perceptron implements precisely $14$ functions.

The truth tables along with some possible values of the parameters $w_1,w_2,b$ are given below.

\begin{table}[h]
\centering
\renewcommand{\arraystretch}{1.2}
\begin{tabular}{c c c c}
\hline
\textbf{Type} & \textbf{Function} & \textbf{Truth Table} $(00,01,10,11)$ & \textbf{Example perceptron parameters $(w_1,w_2,b)$} \\
\hline

Constant  & $0$ & $(0,0,0,0)$ & $(0,0,-1)$ \\

Constant  & $1$ & $(1,1,1,1)$ & $(0,0,1)$ \\

Single-variable  & $x_1$ & $(0,0,1,1)$ & $(1,0,-\tfrac12)$ \\

Single-variable & $\neg x_1$ & $(1,1,0,0)$ & $(-1,0,\tfrac12)$ \\

Single-variable  & $x_2$ & $(0,1,0,1)$ & $(0,1,-\tfrac12)$ \\

Single-variable  & $\neg x_2$ & $(1,0,1,0)$ & $(0,-1,\tfrac12)$ \\

AND  & $x_1 \land x_2$ & $(0,0,0,1)$ & $(1,1,-\tfrac32)$ \\

NAND & $\neg(x_1 \land x_2)$ & $(1,1,1,0)$ & $(-1,-1,\tfrac32)$ \\

OR  & $x_1 \lor x_2$ & $(0,1,1,1)$ & $(1,1,-\tfrac12)$ \\

NOR & $\neg(x_1 \lor x_2)$ & $(1,0,0,0)$ & $(-1,-1,\tfrac12)$ \\

Assymetric & $\neg x_1 \lor x_2$ & $(1,1,0,1)$ & $(-1,1,\tfrac12)$ \\

Assymetric& $\neg x_2 \lor x_1$ & $(1,0,1,1)$ & $(1,-1,\tfrac12)$ \\

Assymetric & $x_1 \land \neg x_2$ & $(0,0,1,0)$ & $(1,-1,-\tfrac12)$ \\

Assymetric & $x_2 \land \neg x_1$ & $(0,1,0,0)$ & $(-1,1,-\tfrac12)$ \\

\hline
\end{tabular}
\caption{\label{table1} All linearly separable Boolean functions on two bits (implementable by a single classical perceptron).}
\end{table}

\section{Witnesses of quantum advantage for every classically implementable perceptron function}

Recall that to witness the predictive advantage of the quantum-IMP versus the classical one, we need to choose the function $f$. Here, we consider all the non-trivial functions that can be implemented in a classical perceptron, from Table \ref{table1}.

\subsection{$f$: Single-variable type}

Let us begin by considering the case when $f$ is single-variable functions. In the main text, we already demonstrated the quantum advantage when $f$ is single-variable $x_1$. Let us now consider the case when $f$ is the single-variable function  $\neg x_1$ [c.f. Table \ref{table1}]. We again conclude that in this example $\mathcal{K}_{z_s^{\mathrm{cor}}}$ is given by  $\mathcal{K}_{z_s^{\mathrm{cor}}}=\{y_{1,s},y_{2,s}|\neg y_{1,s}=z_s^{\mathrm{cor}}=\neg x_{1,s}\}$. Consequently, the prediction witness in this case is given by
\begin{eqnarray}\label{psucapp1}
    p_{\mathrm{cor}}=\frac{1}{8}\sum_{s=0,1}\sum_{x_{1,s}=0,1}\sum_{y_{1,s}=0,1\atop y_{1,s}= x_{1,s}}p(y_{1,s}|s,x_{1,s},x_{1,s}).
\end{eqnarray}
Again, using the classical strategy described in the manuscript, that is, sending the bit corresponding to either of $s=0,1$ and then randomly guessing the other bit, we obtain $\beta_C=\frac{3}{4}$. Consider now the same quantum states as in the main text in Eq. \eqref{qstates} corresponding to inputs $x_{i,s}$ as
\begin{eqnarray}\label{qstatesapp}
\ket{\psi_{x_{i,0},x_{i,1}}}=\sigma_x^{x_{i,0}} \sigma_z^{x_{i,1}}
\left(
\cos\frac{\pi}{8}\,\ket{0}
+
\sin\frac{\pi}{8}\,\ket{1}
\right)
\end{eqnarray}
Using the same measurements as Eq. \eqref{mea1}, that is, for $(j=0,1)$
\begin{eqnarray}
\mathcal{M}_{i,j}^0&=&\proj{ i}\otimes\I_2,\qquad i=0,1\nonumber\\    \mathcal{M}_{0,j}^1&=&\proj{+}\otimes\I_2,\quad \mathcal{M}_{1,j}^1=\proj{-}\otimes\I_2,
\end{eqnarray}
we obtain the quantum value of $p_{\mathrm{cor}}$ \eqref{psucapp1} to be $\beta_Q=1/2(1+1/\sqrt{2})$. Similarly, we can obtain the expression for predictive capability of the perceptron when the other two single-variable functions are considered. We state them below as a table. We also state the quantum measurements for obtaining the quantum value $\beta_Q$, with the same state \eqref{qstatesapp}.

\begin{table}[h]
\centering
\renewcommand{\arraystretch}{1.5}
\begin{tabular}{|c|c|c|c|c|}
\hline
\textbf{Function-Type} & \textbf{Prediction witness} & $\beta_C$ & \textbf{Quantum measurements} & $\beta_Q$ \\
\hline

Single-variable $(x_1)$ 
& \(\displaystyle p_{\mathrm{cor}}=\frac{1}{8}\sum_{s}\sum_{x_{1,0},x_{1,1}} \sum_{y_{1,s}=x_{1,s}} p(y_{1,s}\,|\,s,x_{1,0},x_{1,1})\)
& $\frac{3}{4}$ 
& $\begin{aligned}
\mathcal{M}_{i,j}^0 &= \proj{ i}\otimes \I_2, \quad i=0,1 \\
\mathcal{M}_{0,j}^1 &= \proj{-}\otimes \I_2 \\
\mathcal{M}_{1,j}^1 &= \proj{+}\otimes \I_2
\end{aligned}$
& $\frac{1}{2}\left(1+\frac{1}{\sqrt{2}}\right)$ \\

\hline
Single-variable $(\neg x_1)$ 
& \(\displaystyle p_{\mathrm{cor}}=\frac{1}{8}\sum_{s}\sum_{x_{1,0},x_{1,1}} \sum_{y_{1,s}=x_{1,s}} p(y_{1,s}\,|\,s,x_{1,0},x_{1,1})\)
& $\frac{3}{4}$ 
& $\begin{aligned}
\mathcal{M}_{i,j}^0 &= \proj{i}\otimes \I_2, \quad i=0,1 \\
\mathcal{M}_{0,j}^1 &= \proj{+}\otimes \I_2 \\
\mathcal{M}_{1,j}^1 &= \proj{-}\otimes \I_2
\end{aligned}$
& $\frac{1}{2}\left(1+\frac{1}{\sqrt{2}}\right)$ \\
\hline
Single-variable $(x_2)$ 
& \(\displaystyle p_{\mathrm{cor}}=\frac{1}{8}\sum_{s}\sum_{x_{2,0},x_{2,1}} \sum_{y_{2,s}=x_{2,s}} p(y_{2,s}\,|\,s,x_{2,0},x_{2,1})\)
& $\frac{3}{4}$
& $\begin{aligned}
\mathcal{M}_{i,j}^0 &= \I_2\otimes \proj{ i}, \quad i=0,1 \\
\mathcal{M}_{0,j}^1 &=  \I_2\otimes\proj{-} \\
\mathcal{M}_{1,j}^1 &=  \I_2\otimes\proj{+}
\end{aligned}$
& $\frac{1}{2}\left(1+\frac{1}{\sqrt{2}}\right)$ \\
\hline
Single-variable $(\neg x_2)$ 
& \(\displaystyle p_{\mathrm{cor}}=\frac{1}{8}\sum_{s}\sum_{x_{2,0},x_{2,1}} \sum_{y_{2,s}= x_{2,s}} p(y_{2,s}\,|\,s,x_{2,0},x_{2,1})\)
& $\frac{3}{4}$ 
& $\begin{aligned}
\mathcal{M}_{i,j}^0 &= \I_2\otimes \proj{ i}, \quad i=0,1 \\
\mathcal{M}_{0,j}^1 &=  \I_2\otimes\proj{+} \\
\mathcal{M}_{1,j}^1 &=  \I_2\otimes\proj{-}
\end{aligned}$
& $\frac{1}{2}\left(1+\frac{1}{\sqrt{2}}\right)$ \\
\hline
\end{tabular}
\caption{\label{table2} All single-variable functions implementable in IMP, along with the predictive advantage in the corresponding Q-IMP.}
\end{table}

\subsection{$f$: And/OR Type}

Let us now consider the case when $f$ is AND/OR type. In the main text, we already demonstrated the quantum advantage when $f$ is AND type function. Let us now consider the case when $f$ is OR type function [c.f. Table \ref{table1}]. We can conclude that in this example $\mathcal{K}_{z_s^{\mathrm{cor}}}$ is given by  $\mathcal{K}_{z_s^{\mathrm{cor}}}=\{y_{1,s},y_{2,s}|y_{1,s}\oplus y_{2,s}=z_s^{\mathrm{cor}}= x_{1,s}\oplus x_{2,s}\}$. Here, $a\oplus b$ implies $a+b$ modulus $2$. Consequently, the prediction witness in this case is given by
\begin{eqnarray}\label{psucapp2}
 p_{\mathrm{cor}}=\frac{1}{80}\sum_s \sum_{\mathbf{x}}\sum_{\mathbf{y_s}\atop y_{1,s}\oplus y_{2,s}=x_{1,s}\oplus x_{2,s}} p(\mathbf{y_{s}}|s,\mathbf{x}).
\end{eqnarray}
Let us recall that $\mathbf{y_{s}}=y_{1,s}, y_{2,s}$ and $\mathbf{x}=x_{1,0},x_{2,0},x_{1,1},x_{2,1}$.
Again, using the classical strategy described in the manuscript, that is, sending the bit corresponding to either of $s=0,1$ and then randomly guessing the other bit, we obtain $\beta_C=\frac{13}{80}$. Consider now the same quantum states as in the main text in Eq. \eqref{qstates} corresponding to inputs $x_{i,s}$ as
\begin{eqnarray}\label{qstatesapp}
\ket{\psi_{x_{i,0},x_{i,1}}}=\sigma_x^{x_{i,0}} \sigma_z^{x_{i,1}}
\left(
\cos\frac{\pi}{8}\,\ket{0}
+
\sin\frac{\pi}{8}\,\ket{1}
\right)
\end{eqnarray}
Using the same measurements as Eq. \eqref{meaand},
\begin{eqnarray}
    M_{i,j}^0&=&\proj{i}\otimes\proj{j},\qquad i,j=0,1\nonumber\\
    M_{i,j}^1&=&\proj{+_i}\otimes\proj{+_j},
\end{eqnarray}
we obtain the quantum value of $p_{\mathrm{cor}}$ \eqref{psucapp2} to be $\beta_Q=\frac{1}{40}\left(11+2\sqrt{2}\right)$. Similarly, we can obtain the expression for predictive capability of the perceptron when NAND, NOR type functions are considered. We state them below as a table. Interestingly, for the proposed prediction witness, one can violate the classical bound with the same state \eqref{qstatesapp} and measurements \eqref{meaand} as done above for AND, OR type functions.

\begin{table}[h]
\centering
\renewcommand{\arraystretch}{1.5}
\begin{tabular}{|c|c|c|c|}
\hline
\textbf{Function-Type} & \textbf{Prediction witness} & $\beta_C$ & $\beta_Q$ \\
\hline

AND $(x_1\land x_2)$ 
& 
\(\displaystyle
       p_{\mathrm{cor}}=\frac{1}{80}\sum_s \sum_{\mathbf{x}}\sum_{\mathbf{y_s}\atop y_{1,s}. y_{2,s}=x_{1,s}. x_{2,s}} p(\mathbf{y_{s}}|s,\mathbf{x})
\)
& $\frac{13}{40}$ 
& $\frac{1}{40}\left(11+2\sqrt{2}\right)$ \\

\hline
NAND $(\neg (x_1\land x_2))$ 
& \(\displaystyle p_{\mathrm{cor}}=\frac{1}{80}\sum_s \sum_{\mathbf{x}}\sum_{\mathbf{y_s}\atop y_{1,s}. y_{2,s}=(x_{1,s}. x_{2,s})}p(\mathbf{y_{s}}|s,\mathbf{x})\)
& $\frac{13}{40}$ 
& $\frac{1}{40}\left(11+2\sqrt{2}\right)$ \\
\hline
OR $x_1\lor x_2$ &
\(\displaystyle p_{\mathrm{cor}}=\frac{1}{80}\sum_s \sum_{\mathbf{x}}\sum_{\mathbf{y_s}\atop y_{1,s}\oplus y_{2,s}=x_{1,s}\oplus x_{2,s}} p(\mathbf{y_{s}}|s,\mathbf{x})\)
& $\frac{13}{40}$ 
& $\frac{1}{40}\left(11+2\sqrt{2}\right)$ \\
\hline
NOR $\neg (x_1\lor x_2)$ 
& \(\displaystyle p_{\mathrm{cor}}=\frac{1}{80}\sum_s \sum_{\mathbf{x}}\sum_{\mathbf{y_s}\atop y_{1,s}\oplus y_{2,s}=(x_{1,s}\oplus x_{2,s})} p(\mathbf{y_{s}}|s,\mathbf{x})\)
& \(\displaystyle\frac{13}{40}\) 
& $\frac{1}{40}\left(11+2\sqrt{2}\right)$ \\
\hline
\end{tabular}
\caption{\label{table3} AND/OR type functions implementable in IMP, along with the predictive advantage in the corresponding Q-IMP.}
\end{table}

\subsection{$f$: Assymetric type}

Let us now consider the case where $f$ is an asymmetric type function $(\neg x_1\land x_2)$ [c.f. Table \ref{table1}]. We can conclude that in this example $\mathcal{K}_{z_s^{\mathrm{cor}}}$ is given by  $\mathcal{K}_{z_s^{\mathrm{cor}}}=\{y_{1,s},y_{2,s}|(\neg y_{1,s}). y_{2,s}=z_s^{\mathrm{cor}}= (\neg x_{1,s}).x_{2,s}\}$. Consequently, the prediction witness in this case is the same as AND type function \eqref{psucand} given by
\begin{eqnarray}\label{psucapp2}
 p_{\mathrm{cor}}=\frac{1}{80}\sum_s \sum_{\mathbf{x}}\sum_{\mathbf{y_s}\atop y_{1,s}. y_{2,s}=x_{1,s}. x_{2,s}} p(\mathbf{y_{s}}|s,\mathbf{x}).
\end{eqnarray}
Similarly, we can obtain the expression for predictive capability of the perceptron when the rest of asymmetric type functions are considered. We state them below as a table. Interestingly, for the proposed prediction witness, one can violate the classical bound with the same state \eqref{qstatesapp} and measurements \eqref{meaand} as done above for AND/OR type functions.

\begin{table}[h]
\centering
\renewcommand{\arraystretch}{1.5}
\begin{tabular}{|c|c|c|c|}
\hline
\textbf{Function-Type} & \textbf{Prediction witness} & $\beta_C$ & $\beta_Q$ \\
\hline

Assymetric $(\neg x_1\land x_2)$ 
& 
\(\displaystyle
       p_{\mathrm{cor}}=\frac{1}{80}\sum_s \sum_{\mathbf{x}}\sum_{\mathbf{y_s}\atop y_{1,s}. y_{2,s}=x_{1,s}. x_{2,s}} p(\mathbf{y_{s}}|s,\mathbf{x})
\)
& $\frac{13}{40}$ 
& $\frac{1}{40}\left(11+2\sqrt{2}\right)$ \\

\hline
Assymetric $(x_1\land \neg x_2)$ 
& \(\displaystyle p_{\mathrm{cor}}=\frac{1}{80}\sum_s \sum_{\mathbf{x}}\sum_{\mathbf{y_s}\atop y_{1,s}. y_{2,s}=(x_{1,s}. x_{2,s})}p(\mathbf{y_{s}}|s,\mathbf{x})\)
& $\frac{13}{40}$ 
& $\frac{1}{40}\left(11+2\sqrt{2}\right)$ \\
\hline
Assymetric $\neg x_1\lor x_2$ &
\(\displaystyle p_{\mathrm{cor}}=\frac{1}{80}\sum_s \sum_{\mathbf{x}}\sum_{\mathbf{y_s}\atop y_{1,s}\oplus y_{2,s}=x_{1,s}\oplus x_{2,s}} p(\mathbf{y_{s}}|s,\mathbf{x})\)
& $\frac{13}{40}$ 
& $\frac{1}{40}\left(11+2\sqrt{2}\right)$ \\
\hline
Assymetric $(x_1\lor \neg x_2)$ 
& \(\displaystyle p_{\mathrm{cor}}=\frac{1}{80}\sum_s \sum_{\mathbf{x}}\sum_{\mathbf{y_s}\atop y_{1,s}\oplus y_{2,s}=(x_{1,s}\oplus x_{2,s})} p(\mathbf{y_{s}}|s,\mathbf{x})\)
& \(\displaystyle\frac{13}{40}\) 
& $\frac{1}{40}\left(11+2\sqrt{2}\right)$ \\
\hline
\end{tabular}
\caption{\label{table2} Asymmetric type functions implementable in IMP, along with the predictive advantage in the corresponding Q-IMP.}
\end{table}




\end{document}